\documentclass[10pt, conference, letterpaper]{IEEEtran}
\IEEEoverridecommandlockouts
\usepackage{bbm}
\usepackage{listings}
\usepackage{cite}
\usepackage{url}
\usepackage{xcolor}
\usepackage{graphicx}
\usepackage{epstopdf}
\usepackage[cmex10]{amsmath}
\usepackage{cases}
\usepackage{amsfonts}
\usepackage{amssymb}
\usepackage{amsthm}
\usepackage{mathrsfs}
\usepackage{verbatim}
\usepackage{algorithm}
\usepackage{algpseudocode}
\usepackage{setspace}
\usepackage{bm}%
\usepackage{stfloats}
\usepackage{enumerate}

\theoremstyle{plain}
\usepackage{multirow}
\newtheorem{theorem}{Theorem}
\newtheorem{remark}{Remark}

\newtheorem{definition}{Definition}

\usepackage{array}

\usepackage{tabularx}
\usepackage{booktabs}  
\usepackage{multirow}  
\usepackage[letterpaper, left=0.625in, right=0.625in, top=0.75in, bottom=0.98in]{geometry}

\usepackage[caption=false,font=normalsize,labelfont=sf,textfont=sf]{subfig}

\usepackage{tcolorbox}
\usepackage{pifont}

\usepackage[caption=false,font=normalsize,labelfont=sf,textfont=sf]{subfig}

\begin{document}

\title{Semantics-Aware Communication: \\A Differentiated Allocation Perspective}

\author{Fangming~Zhao\IEEEauthorrefmark{1}, Nikolaos Pappas\IEEEauthorrefmark{2}, and Howard H.~Yang\IEEEauthorrefmark{1}\\
    \IEEEauthorrefmark{1} ZJU-UIUC Institute of Zhejiang University, Haining, China\\
    \IEEEauthorrefmark{2} Department of Computer and Information Science, Linköping University, Linköping, Sweden
    \thanks{This paper was supported in part by the National Natural Science Foundation of China under Grant 62671550, in part by the Yunnan Science and Technology Program under Grant 202602AD080017, and in part by the National Key R\&D Program of China under Grant 2024YFE0200700. The work of Nikolaos Pappas was supported by ELLIIT, and the European Union (6G-LEADER, 101192080, and MAGIC-6G, 101292933). (\textit{Corresponding Author: Howard H. Yang.})}}

\maketitle
\thispagestyle{empty}

\begin{abstract}
We study the joint optimization of timeliness and reliability in semantics-aware Wireless Networked Control Systems (WNCS) under computational resource constraints. The sampled data are categorized into regular and critical tasks based on the semantic states, facilitating differentiated resource allocation. Task-aware Age of Actuated Information (AoAI) and Cost of Missing Actuation (CoMA) are used to characterize the task-level freshness and the reliability penalty of missed actuations, respectively. By modeling the controller as a discrete-time multi-rate Geo/D/C/C queue, we evaluate the performance of regular and critical tasks, the latter imposing higher computational demands. Results confirm that differentiated resource allocation across heterogeneous tasks effectively improves the actuation reliability of critical tasks in severely constrained environments.
\end{abstract}

\begin{IEEEkeywords}
Semantics-aware, multi-service loss system, Age of Actuated Information (AoAI), CoMA.
\end{IEEEkeywords}

\section{Introduction}
The Age of Information (AoI) serves as a fundamental metric to quantify information freshness in time-critical networks \cite{INFOCOM2012:AoIstart, JSACAoIsurvey}. By capturing the time elapsed since the generation of the most recently received update, AoI facilitates the design of communication systems with timeliness constraints \cite {ZhaoTWC2024}. However, AoI remains a transmission-centric metric that falls short in WNCS. In such systems, the intrinsic utility of data depends heavily on actuation rather than mere packet delivery. Recognizing that system efficacy is governed by task semantics, it is imperative to shift the focus toward semantics-aware communication \cite{luoNikosTCOMInvite}. Under this paradigm, transmission timing and resource allocation are explicitly designed to meet the underlying execution objectives of the receiver.

Within this transmission-to-execution framework, the AoAI metric characterizes the temporal mismatch between the current system state and the most recently executed control action \cite{Nikos:AoA}. By shifting the focus from information reception to action execution, AoAI provides a direct link between communication timeliness and control effectiveness. Prior studies on AoAI are largely limited to single-stream settings, overlooking the heterogeneity of different tasks. In real-time control, dropping a critical packet can severely degrade system stability, necessitating differentiated treatment in system design. Therefore, we need to evaluate the AoAI for each different class independently, which is referred to as task-aware AoAI.

Furthermore, for rare yet highly critical events, AoAI alone is insufficient to characterize their true impact. For infrequent tasks, AoAI alone cannot distinguish between rare task generation and unreliable execution after generation. Once an actuation demand is generated, the consequence of failing to execute it should be determined by its task-level significance rather than be directly weighted by generation probability. We therefore define CoMA as the conditional expected miss cost given task generation. Note that the generation probabilities of the task classes still affect the conditional miss probabilities indirectly through the offered traffic loads and resource contention. 
By combining the AoAI of each task and CoMA for execution reliability, we can capture the multidimensional requirements of a WNCS via a bi-objective optimization framework\cite {li2026MPCareto}.

Despite recent advancements in semantics-aware system design\cite{luoNikosTCOMInvite}, most existing works primarily emphasize communication resource allocation while overlooking the joint effects of communication and computation. In systems constrained by limited processing capacity, the implicit interaction between execution reliability and competition for computing resources remains underexplored. To solve these issues, we develop a framework that jointly models: $(i)$ task-oriented sample transmission over imperfect uplinks; $(ii)$ the dynamic allocation and release of computational resources at the controller; $(iii)$ heterogeneous resource requirements of different task types. The main contributions of this work are summarized as follows:
\begin{itemize}
    \item We model a semantics-aware WNCS constrained by limited computational resources. Specifically, a multi-rate Geo/D/C/C queueing system is developed to capture the computational resource dynamics at the controller. We compare its blocking probability against both the \mbox{ Geo/Geo/C/C} model and the continuous-time Erlang formula. The results reveal that the blocking probability in discrete-time loss queueing systems is not strictly insensitive to the service distribution. Nevertheless, a stochastic service model can still be employed as an approximation for deterministic service, thereby mitigating the curse of dimensionality caused by multi-rate deterministic models.
    \item We analyze two semantics-aware metrics, namely task-aware AoAI and CoMA, which characterize the timeliness and reliability of task execution, respectively. Our analysis reveals strong competition for limited computational resources among heterogeneous task flows, whereby the execution of one task class can affect the performance of the others. Numerical results demonstrate that differentiated resource allocation can improve resource efficiency and enhance the execution reliability of critical tasks in severely resource-constrained WNCS.
\end{itemize}

\section{System model}
We consider a WNCS that integrates communication and computation. As illustrated in Fig. \ref{SystemModelBiChannel}, the system comprises multiple modules: a sensor to generate sampled data, an uplink channel for data transmission, an edge controller to process tasks and make decisions, and an actuator to execute the task commands. In the following, we introduce the mathematical modeling of each module, which allows for formulating the performance metrics: the task-aware AoAI and the CoMA.
\begin{figure}
    \centering
    \includegraphics[height=3.6cm,width=8cm]{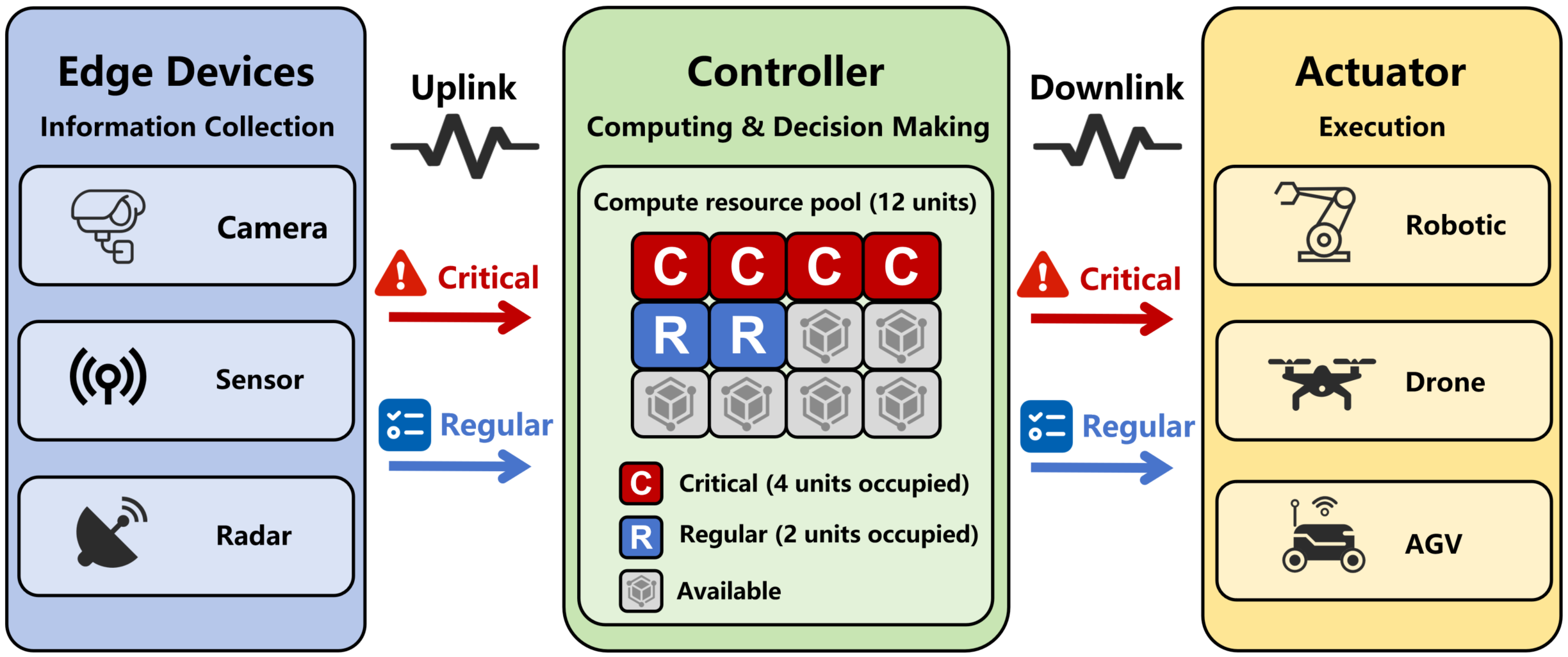}
    \caption{Semantics-aware wireless networked control systems architecture.}
    \label{SystemModelBiChannel}
    \vspace{-0.5cm}
\end{figure}

\subsection{Wireless Networked Control System Modeling}

\subsubsection{\textbf{Semantics-Aware Information Generation}} We consider a discrete-time framework where the sensor generates sampling data to support distinct perception and control functionalities. To optimize the utilization of wireless resources while guaranteeing system stability, the generated data is classified into two categories: regular information and critical information. We aim to implement differentiated resource allocation based on the semantics of the information.

Specifically, the semantic value of the data is dynamically evaluated based on the local physical state. A regular packet is generated when the discrepancy exceeds a predefined regular threshold but remains within a safe margin. In contrast, a critical packet is triggered when the discrepancy violates a critical safety limit. Let $g_1$ and $g_2$ denote the overall probabilities of generating regular information and critical information in a given time slot, respectively, where $g_1+g_2\leq 1$. The remaining probability $1-g_1-g_2$ accounts for the idle time slots where the physical state is well-predicted and no data is generated. To streamline the presentation, we refer to the control commands associated with the regular and critical information streams as Task 1 and Task 2, respectively, indexed by $i\in\mathcal{T}\triangleq\{1,2\}$.

\subsubsection{\textbf{Uplink Communication Model}} Once a packet is generated, the sensor attempts to transmit it to the controller over the uplink wireless channel. Given a generated class $i$ packet, the sensor accesses the uplink with probability $\eta_i$. The uplink channel experiences small-scale fading, where the fading envelope $|h_u|$ follows a Nakagami-$m$ distribution with a shape parameter $m$ and a normalized average power $\mathbb{E}[|h_u|^2] = 1$. Additionally, it suffers from large-scale path loss modeled as $d_u^{-\alpha_u}$, where $d_u$ is the distance from the sensor to the controller and $\alpha_u$ is the path-loss exponent.
Consequently, the received signal-to-noise ratio (SNR) $\gamma_{u,i}$ at the controller is $\gamma_{u,i} = \frac{P_{T,i}|h_u|^2 d_u^{-\alpha_u}}{\sigma^2}$, where $P_{T,i}$ is the transmit power of the sensor, and $\sigma^2$ is the noise power at the controller. If $\gamma_{u,i}$ falls below a decoding threshold $\bar{\gamma}$, the uplink packet is lost. The uplink transmission success probability for task $i$, denoted by $p_{u,i}$, is given by $p_{u,i}=\mathbb{P}(\gamma_{u,i}\geq\bar{\gamma})$. 

However, combating channel impairments to guarantee the required $p_{u,i}$ requires sufficient transmit power $P_{T,i}$, which inherently drains the limited battery life of the edge device. To ensure the practical feasibility of the proposed strategy, we impose a long-term average transmit-power budget $\bar{P}$ at the edge device. Accordingly, the differentiated transmission parameters for heterogeneous tasks must satisfy $\sum_{i\in\mathcal{T}} g_i\eta_i P_{T,i}\leq \bar{P}$.

\subsubsection{\textbf{Controller with Limited Computational Resource}}
The controller acts as a computing node with limited computational capacity, comprising $C$ parallel and independent computation units. Upon successfully decoding an uplink packet, the receiver performs a lightweight header parsing to identify its semantic type (Task 1 or Task 2). Based on this identification, the controller allocates computation resources following a semantics-aware resource allocation policy:\footnote{We differentiate resource allocation as follows: regular tasks use semantic priors for efficient recovery, while critical tasks require high-fidelity reconstruction, thereby requiring greater computational effort.}
\begin{itemize}
    \item \textbf{Task 1}: requiring $N_1=1$ computation unit. The deterministic service time is $D_{C,1}$ slots.
    \item \textbf{Task 2}: requiring $N_2=N$ computation units ($N>1$) for complex calculations. The assigned $N$ units process the task in a synchronous parallel manner, with a deterministic service time of $D_{C,2}$ slots.
\end{itemize}

All computation units are non-preemptive. Once allocated to a task, they remain occupied until the computation is completed. We adopt an admission-before-release convention at each slot boundary. Admission is determined using the pre-release resource occupancy. If computational resources are insufficient upon task arrival, the task is discarded without entering the buffer. Therefore, the dynamics of the computational resource pool can be modeled as a multi-rate Geo/D/C/C queueing system.

\subsubsection{\textbf{Downlink Communication Model}}
The traffic link of actuation commands from the controller to the actuator is assumed to be error-free, and only considers constant delay $D_{T,i}$. This assumption is well justified by the highly asymmetric nature of WNCS links: the controller typically operates with an ample power budget and a continuous energy supply, enabling it to employ advanced channel-adaptive schemes to guarantee a near-100\% successful delivery rate. Consequently, the primary system bottlenecks are dominated by the uplink transmission and the computation queueing.

\subsubsection{\textbf{Actuation Model}}
The actuation system is deployed at the end device. Successfully processed task commands are immediately executed. If a regular task is dropped due to uplink transmission failure or computational resource limits, the actuator maintains the last valid command.\footnote{Our differentiated allocation strategy exploits the disparate error tolerances inherent in semantics-aware control systems. While packet loss in critical tasks causes catastrophic failure, regular tasks possess semantic redundancy. Leveraging Zero-Order Hold mechanisms\cite{luoNikosTCOMInvite} at the actuator, the system exhibits degradation against the dropping of regular tasks. This property justifies sacrificing regular tasks to guarantee the reliability of critical tasks.} While ensuring continuous physical input, holding stale commands degrades system stability. This directly motivates our use of the AoAI and the CoMA to quantify the penalty of state deviations.

\subsection{Performance Metrics}
In time-critical systems, the AoI metric is commonly adopted as a key performance metric to quantify information freshness. According to the original definition, the AoI metric does not differentiate between the nature of data flows and focuses solely on the transmission process from source to destination. As such, it does not adequately capture the heterogeneity and computation execution processes in task-oriented networks.

In this work, we further extend two performance metrics: the AoAI and the CoMA \cite{Nikos:AoA}. Given the presence of different types of tasks, we evaluate these metrics separately for each task to capture their distinct characteristics and their impact on performance. The definitions are given as follows.
\begin{definition}
\textit{(Task-aware Age of Actuated Information)} Consider a system that involves $n$ types of task commands, task $i$, $i\in\{1,2,...,n\}$. Let ${A_i(t_{i,j})}$, ${j \ge 1}$ be the sequence of generation times of information packets belonging to task type $i$ that were successfully executed by the actuator. The index of the most recently executed task $i$ update by time $t$ is
\begin{equation}
    J_i(t) \triangleq \max \{j:t_{i,j}<t\}. 
\end{equation}
Accordingly, the generation time of the most recently executed task $i$ update is $A_i(t)\triangleq A_i(t_{i,J_i(t)})$. Therefore, the AoAI for task type $i$ at time $t$ is defined as  
\begin{equation}
    \Delta_i(t)\triangleq t-A_i(t).
\end{equation}
The AoAI of task class $i$ is reset exclusively upon the successful execution of a task $i$ commands, with the reset value equal to the corresponding generation-to-actuation delay.
\end{definition}

Following Definition 1, the time-average function of the task-aware AoAI metric (i.e., time-average AoAI) is given by
   \begin{equation}
\bar{\Delta}_i\triangleq\limsup_{T\to \infty}\frac{1}{T}\mathbb{E}\left[\sum_{t=1}^{T}\Delta_i(t)\right].
   \end{equation}

The time-average AoAI characterizes the long-term timeliness of operator actuation. Moreover, to evaluate the reliability of transmission over imperfect channels and the performance degradation caused by insufficient computational resources, we extend the general penalty-based performance metric, the Cost of Missing Actuation, to the multi-task case to denote the cost of discarding a task command. Specifically, a penalty of $\omega_i$ is incurred for each unsuccessfully executed Task $i\in\mathcal{T}$. 
\begin{definition}
  Cost of Missing Actuation (CoMA): The aggregate long-term expected loss cost per generated task of each class is given by
  \begin{equation}
      \mathcal{C}_{\mathrm{MA}}
\triangleq
\sum_{i\in\mathcal{T}}
\limsup_{T\rightarrow\infty}
\frac{
\mathbb{E}\left[
\sum_{t=1}^{T}L_i(t)\mathbb{I}_i(t)
\right]
}{
\mathbb{E}\left[
\sum_{t=1}^{T}G_i(t)
\right]
},
\end{equation}
  where $L_i(t)$ denotes the cost associated with the failure of a class $i$ task at time $t$, $\mathbb{I}_i(t)\in\{0, 1\}$ is an indicator function, and $\mathbb{I}_i(t)=1$ indicates the class $i$ task loss due to ($i$) no transmission is attempted; ($ii$) uplink decoding fails; or ($iii$) the controller lacks sufficient computing resources, $G_i(t)\in \{0,1\}$ indicates whether a class-$i$ task is generated at time $t$.
\end{definition}

\section{Performance analysis}
By leveraging a discrete-time graphical approach \cite{JSAC2025:AoIEH}, the time-average AoAI for task $i$, denoted by $\Delta_i$, can be formulated as
\begin{equation}\label{def:AoA}
\bar{\Delta}_i=\frac{\mathbb{E}[X_{i}^2]+\mathbb{E}[X_{i}]}{2\mathbb{E}[X_{i}]}+\mathbb{E}[D_{C,i}]+\mathbb{E}[D_{T,i}],
\end{equation}
where random variable $X_i$ represents the time interval between two consecutive executions of task $i$. Let $\Gamma$ denote the number of idle computing units observed by an arriving task at its admission epoch. By using the squared coefficient of variation $c^2_{X_i}\triangleq \mathrm{Var}(X_i)/\mathbb{E}[X_i]^2$, the analytical expressions for these metrics are established in Theorem 1.
\begin{theorem}
 Time-average AoAI of task $i$ can be expressed as
  \begin{align}\label{eq:AoA:main}
      \bar{\Delta}_i 
      &= \frac{1+c^2_{X_i}}{2g_i \eta_i  p_{u,i} \mathbb{P}(\Gamma\geq N_i)}+\frac{1}{2}+D_{C,i} + D_{T,i},\\
      &\approx \frac{1}{g_i \eta_i  p_{u,i} \mathbb{P}(\Gamma\geq N_i)}+D_{C,i} + D_{T,i},
   \end{align}
where $c^2_{X_i}$ can be evaluated via a first-passage-time analysis. The CoMA can be given by
\begin{align}\label{eq:CoMA:main}
    \mathcal{C}_{\mathrm{MA}}\!=\!\sum_{\small i \in \{1,2\}} \omega_i (1-\eta_i p_{u,i}\mathbb{P}(\Gamma\geq N_i)),
\end{align}
where $p_{u,i} = \frac{\int_{m\psi_i}^{\infty} t^{m-1} e^{-t} dt}{\int_{0}^{\infty} t^{m-1} e^{-t} dt}$, and $\psi_i={\frac{\bar{\gamma} \sigma^2 d_u^{\alpha_u}}{P_{T,i}}}$.
\end{theorem}
\begin{IEEEproof}
Let $r_i=g_i \eta_i  p_{u,i} \mathbb{P}(\Gamma\geq N_i)$. 
Then, $\mathbb{E}[X_{i}]=1/r_i$ and $\mathbb{E}[X^2_{i}]=(1+c^2_{X_i})\mathbb{E}[X_{i}]^2$. Substituting these moments and the constant computation and transmission delays into \eqref{def:AoA} yields \eqref{eq:AoA:main}.
The successful execution process of class-$i$ tasks is approximated as a Bernoulli process yielding $ c^2_{X_i}\approx1-r_i$.  The CoMA is formulated by taking one minus the joint probability of a transmission attempt, successful uplink decoding, and sufficient computation, scaled by task penalties $L_i(t)\!=\!\omega_i$. The transmission success probability $p_{u,i} = \mathbb{P}(\gamma_{u,i} \geq \bar{\gamma})$ is reformulated as $\mathbb{P}(|h_u|^2 \geq \psi_i)$, which is then solved via a variable substitution of the Gamma distribution.
\end{IEEEproof}
\begin{remark}
   As a baseline for comparison, the time-average AoI from the sensor to the controller can be expressed as
    \begin{equation}\label{eq:AoI}
   \bar{\Delta}\!=\!\frac{1}{g_1\eta_1 p_{u,1}+g_2\eta_2p_{u,2}}.   
\end{equation}
\end{remark}
From \eqref{eq:AoA:main}–\eqref{eq:CoMA:main}, we observe that the performance of different task flows is coupled through the probability of sufficient computational resources. Specifically, the uplink access probability of task 1 $\eta_1$ influences computational resource occupancy, which in turn affects computational resource availability and timeliness for task 2. Therefore, in resource-constrained regimes, the differentiated treatment of heterogeneous semantic streams is crucial to maximizing the overall system efficacy.

Furthermore, calculating these metrics requires deriving the probability that the controller possesses sufficient computational units upon the arrival of Task $i$, $\mathbb{P}(\Gamma\geq N_i)$. We then model the computation resource allocation and release.

\subsection{Computation Resource Dynamics in Deterministic Service}
The allocation and release of computation units give rise to system dynamics. When the computation service time is deterministic, the memoryless property of task departure is strictly broken. To capture the remaining processing time of each active task, we model the system evolution using a $(D_{C,1}+D_{C,2})$-dimensional
discrete-time Markov chain (DTMC). 

Let $\mathbf{v}_i = [v_{i,1}, v_{i,2}, \dots, v_{i,D_{C,i}}]^T \in \{0, 1\}^{D_{C,i}}$ denote the execution pipeline for Task $i \in \{1, 2\}$, where $D_{C,i}$ is the deterministic required service time (in slots). The binary element $v_{i,k} = 1$ indicates that there is currently a Task $i$ instance with exactly $k$ slots of remaining execution time, and $v_{i,k} = 0$ otherwise. The system state is then defined as a composite vector $\boldsymbol{s} = (\mathbf{v}_1, \mathbf{v}_2)$. The number of currently active tasks of type $i$, denoted by $n_i(\mathbf{v}_i)$, is the $\ell_1$-norm of the vector, i.e., $n_i(\mathbf{v}_i) = \|\mathbf{v}_i\|_1 = \sum_{k=1}^{D_{C,i}} v_{i,k}$. Consequently, the valid state space under the capacity constraint is given by: $\boldsymbol{\Omega}_{\text{Det}} \!=\!\left\{ (\mathbf{v}_1, \mathbf{v}_2)\!\in\!\{0,1\}^{D_{C,1}}\!\times\!\{0,1\}^{D_{C,2}} \;\middle|\!\!\; \|\mathbf{v}_1\|_1 \!+\! N \|\mathbf{v}_2\|_1 \!\leq\! C \right\}.$

An analytically tractable upper bound on the number of feasible states $|\boldsymbol{\Omega}_{\text{Det}}|$ can be derived as:
\begin{equation}\label{eq:D:num}
 |\boldsymbol{\Omega}_{\text{Det}}| \!\!\leq\!\! \sum_{n_2=0}^{\min\left(D_{C,2}, \lfloor C/N \rfloor\right)} \sum_{n_1=0}^{\min\left(D_{C,1}, C - N n_2\right)}\!\!\!\binom{D_{C,1}}{n_1} \!\binom{D_{C,2}}{n_2}\!, 
\end{equation}
where the binomial coefficient $\binom{D_{C,i}}{n_i}$ represents the number of possible execution configurations that contain exactly $n_i$ active tasks. Although the bound of capacity-feasible configurations may include unreachable states, its cardinality provides an upper bound that reflects the combinatorial growth of the deterministic-service state space.

The state transition is driven by the deterministic aging of existing tasks and the admission of newly arrived tasks. Let $a_i(t) \in \{0,1\}$ be the admission indicator for Task $i$ at time slot $t$. As time progresses from slot $t$ to $t+1$, all active tasks advance one step towards completion. A task with $v_{i,1}(t) = 1$ completes its execution and departs the system, releasing its resources. Simultaneously, a newly admitted task enters the pipeline with a full remaining time of $D_{C,i}$. Mathematically, this physical evolution is modeled as a deterministic shift operation $\mathcal{F}(\mathbf{v}_i, a_i)$:
\begin{equation}
v_{i,k}(t+1) = \begin{cases} v_{i, k+1}(t), & \text{for } 1 \leq k \leq D_{C,i} - 1, \\ a_i(t), & \text{for } k = D_{C,i}. \end{cases}
\end{equation}
The admission probabilities $P_{\mathcal{A}}(a_1, a_2 \mid \!\boldsymbol{s})$ follow the same logic as the memoryless case, heavily depending on the current available capacity $C - \|\mathbf{v}_1\|_1 - N \|\mathbf{v}_2\|_1$. Let $\mathcal{A} = \{(0,0), (1,0), (0,1)\}$ denote the feasible admission action space. The state transition probability $P_{\boldsymbol{s}}^{\boldsymbol{s}'}$ from current state $\boldsymbol{s} = (\mathbf{v}_1, \mathbf{v}_2)$ to the next state $\boldsymbol{s}' = (\mathbf{v}'_1, \mathbf{v}'_2)$ is concisely expressed as:
\begin{align}\notag
\sum_{(a_1, a_2) \in \mathcal{A}}\!\!\!P_{\mathcal{A}}(a_1, a_2 \mid \boldsymbol{s}) \mathbb{I}\Big(\mathbf{v}'_1 \!=\! \mathcal{F}(\mathbf{v}_1, a_1)\Big)\mathbb{I}\Big(\mathbf{v}'_2 \!=\! \mathcal{F}(\mathbf{v}_2, a_2)\Big),
\end{align}
where the indicator function $\mathbb{I}(\cdot)$ ensures that the transition probability is strictly zero unless the target state exactly matches the deterministic shift execution pipeline, and:
\begin{align*}
   & P_{\mathcal{A}}(a_1, a_2 \mid \boldsymbol{s}) = \notag\\
   & \begin{cases}
      \!  g_1\eta_1 p_{u,1}  \mathbb{I}\big(n_1(\boldsymbol{s}) \!+\! N n_2(\boldsymbol{s})\!+\! 1 \leq C\big), & \text{if } (a_1, a_2) \!=\! (1, 0), \\
      \!  g_2\eta_2 p_{u,2}\mathbb{I}\big(n_1(\boldsymbol{s}) \!+\!N n_2(\boldsymbol{s})\!+\!N \leq C\big), & \text{if } (a_1, a_2)\!=\! (0, 1), \\
      \!  1 - P_{\mathcal{A}}(1, 0 \mid \boldsymbol{s}) - P_{\mathcal{A}}(0, 1 \mid \boldsymbol{s}), & \text{if } (a_1, a_2) \!=\! (0, 0).
    \end{cases}
    \label{eq:admission_prob}
\end{align*}

Next, we solve for the steady-state probability vector $\mathcal{S}$. To achieve this, we traverse the state transition space via Breadth-First Search (BFS) without exhaustive enumeration, and subsequently solve the global balance equation $\mathcal{S} (\mathbf{P} - \mathbf{I}) = \mathbf{0}$ subject to the normalization constraint $\mathcal{S} \mathbf{1} = 1$. However, it is evident from \eqref{eq:D:num} that the size of the state space increases combinatorially with both the total amount of computational resources and the computation delay of each task. As a result, exact numerical evaluations that require constructing the full steady-state transition matrix and employing sparse matrix solvers remain computationally tractable only at limited scales, such as $C = 20$ and $D_{C,i} = 10$.

Alternatively, by approximating the deterministic service completions as a Bernoulli arrival process, the resource release process leverages the memoryless property. This allows us to bypass tracking execution-time states, successfully reducing the DTMC from $(D_{C,1}+D_{C,2})$ dimensions to a strictly $2$-dimensional model, which is presented in the next subsection.

\subsection{Computation Resource Dynamics in Random Service}

When the computation service time follows a geometric distribution, the resource pool can be modeled as a multi-rate Geo/Geo/C/C queueing. We can model this process as a two-dimensional DTMC. The state is defined as $\boldsymbol{s}=(n_1, n_2)$, representing the number of active Task 1 and Task 2 instances, respectively. Let $\boldsymbol{\Omega}_{\text{Geo}}$ be the state space of the Markov chain, given by
$\boldsymbol{\Omega}_{\text{Geo}}=\{(n_{1},n_{2}) \mid n_{1}, n_{2}\in \mathbb{N}, n_{1}+Nn_{2}\leq C\}.$ 
The total number of valid states can be derived as
\begin{align}
|\boldsymbol{\Omega}_{\text{Geo}}|\!&=\!\tfrac{1}{2}\left(C\!+\!1\!+\!\left(C\!+\!1\!-\!N\!\times\!\left\lfloor\tfrac{C}{N} \right\rfloor\right)\right)\!\!\left(\left\lfloor\tfrac{C}{N}\right\rfloor\!+\!1\right)\!\approx\! \tfrac{C^2}{2N}.
\end{align}

The state transition from current state $(n_1, n_2)$ to the next state $(\nu, \varphi)$ depends on two independent sub-processes in a single time slot: ($i$) the departure of completed tasks, and ($ii$) the admission of newly arrived tasks based on available resources. First, we define the probability mass function of the number of remaining tasks. Since each task independently completes with probability $\mu_i=1/D_{C,i}$, the number of remaining Task $i$ instances, denoted by $\kappa$, follows a Binomial distribution $B(n_i, 1-\mu_i)$. We define departure function as:
\begin{equation}
\begin{small}
B(n_i, \mu_i \mid \kappa) = \begin{cases} \binom{n_i}{\kappa} (1-\mu_i)^{\kappa} \mu_i^{n_i-\kappa}, & 0 \leq \kappa \leq n_i, \\ 0, & \text{otherwise}. \end{cases}
\end{small}
\end{equation}
\par
Second, we calculate the task admission probabilities. A newly arrived task is admitted only if the system has sufficient remaining capacity. Let $a_1 \in \{0, 1\}$ and $a_2 \in \{0, 1\}$ denote the number of admitted Task 1 and Task 2 instances, respectively. Given the current state $(n_1, n_2)$, the joint admission probability $P_{\mathcal{A}}(a_1, a_2 \!\mid\! n_1, n_2)$ is defined as:
\begin{equation*}
\begin{split}
\begin{cases}
\!P_{\mathcal{A}}(1,0 \mid n_1,n_2) \!=\! g_1\eta_1 p_{u,1}  \mathbb{I}(n_1 + Nn_2 + 1 \leq C) \\ 
\! P_{\mathcal{A}}(0,1 \mid n_1,n_2) \!=\! g_2\eta_2 p_{u,2}  \mathbb{I}(n_1 + Nn_2 + N \leq C) \\ 
\!P_{\mathcal{A}}(0,0 \mid n_1,n_2) \!=\! 1 \!-\! P_{\mathcal{A}}(1,0 \mid n_1,n_2) \!-\! P_{\mathcal{A}}(0,1 \mid n_1,n_2), \end{cases}
\end{split}
\end{equation*}
where $\mathbb{I}(\cdot)$ is the indicator function that returns $1$ if the condition is true, and $0$ otherwise. Finally, the target next state $(\nu, \varphi)$ is simply the sum of the remaining tasks and the newly admitted tasks. By applying the law of total probability over all possible admission events, the unified state transition probability $P_{(n_{1},n_{2})}^{(\nu,\varphi)}$ is concisely expressed as:
\begin{equation}\notag
 \sum_{(a_1, a_2) \in \mathcal{A}}\!\!\!\!P_{\mathcal{A}}(a_1, a_2 \!\mid\! n_1, n_2) B(n_1, \mu_1  \!\mid\! \nu \!-\! a_1) B(n_2, \mu_2  \!\mid\! \varphi \!-\! a_2). 
\end{equation}

The state transition probability matrix is denoted by $\boldsymbol{P}\triangleq\left[P_{(n_{1},n_{2})}^{(\nu,\varphi)}\right]\in \mathbb{R}^{{|\boldsymbol{\Omega}_{\text{Geo}|}\times |\boldsymbol{\Omega}_{\text{Geo}}|}}$. Leveraging the structural properties of the proposed DTMC's transition matrix, we can efficiently solve the system by employing a matrix-geometric approach. Detailed derivations are provided in Appendix \ref{proof:AppendixA}.

Further, we also present the steady state probability results for the continuous-time M/G/C/C model, which is often used as an approximation to the discrete time queueing. 

{\textbf{Approximation:} (Erlang's loss formula) The steady-state distribution of the number of tasks in the system for an M/G/C/C model is given by \cite{ErlangBook}:
\begin{equation}
\mathcal{S}_{n_{1},n_{2}}\!=\!\frac{\frac{\rho_1^{n_1}}{n_1!} \frac{\rho_2^{n_2}}{n_2!}}{\sum\limits_{(n_1,n_2)\in  \boldsymbol{\Omega}_\text{C}} \frac{\rho_1^{n_1}}{n_1!} \frac{\rho_2^{n_2}}{n_2!}},~\rho_i = \frac{g_i\eta_ip_{u,i}}{\mu_i},~i\in\{1,2\}.
\end{equation}
}

This result can be used as a tight approximation under light traffic, because the system dynamics are largely governed by transitions between adjacent states, with multi-server transitions for different tasks occurring with negligible probability. Consequently, Erlang's loss formula provides a tractable approximation for characterizing system behavior, albeit at the expense of analytical accuracy.

Based on the steady-state distribution analysis, the probability that task $i$ is successfully allocated a sufficient number of computation units can be obtained as:
\begin{equation}\label{eq:ProC}
\mathbb{P}(\Gamma\geq N_i) =\sum_{\boldsymbol{s} \in \boldsymbol{\Omega}}\mathcal{S}(\boldsymbol{s}) \cdot \mathbb{I}\Big(n_1(\boldsymbol{s}) + N n_2(\boldsymbol{s}) + N_i \leq C\Big),
\end{equation}
where $\Omega$ denotes the state space of the general queueing model.

Fig. 2 compares the blocking probabilities when task $i$ arrives at the controller, i.e., $1-\mathbb{P}(\Gamma\geq N_i)$, for the Geo/D/C/C, Geo/Geo/C/C, and M/G/C/C queues, all modeled under identical traffic loads. The exact Geo/D/C/C analysis closely matches the simulations when the state space is moderate in size. Notably, the Geo/D/C/C queue exhibits the lowest blocking probability over the illustrated range, as deterministic service times eliminate service-time variance. Consequently, the Geo/Geo/C/C and M/G/C/C models serve as valid analytical surrogates, with negligible approximation gaps in the light-traffic regime. The Geo/Geo/C/C queue provides a tight surrogate for evaluating the performance of Task 2, effectively bypassing the curse of dimensionality in the multi-rate Geo/D/C/C analysis.

\begin{figure}[t]
\captionsetup[subfigure]{font=footnotesize, labelfont=sf, textfont=sf}
    \centering
    \subfloat[Task 1]
{\includegraphics[height=3.6cm,width=4.65cm]{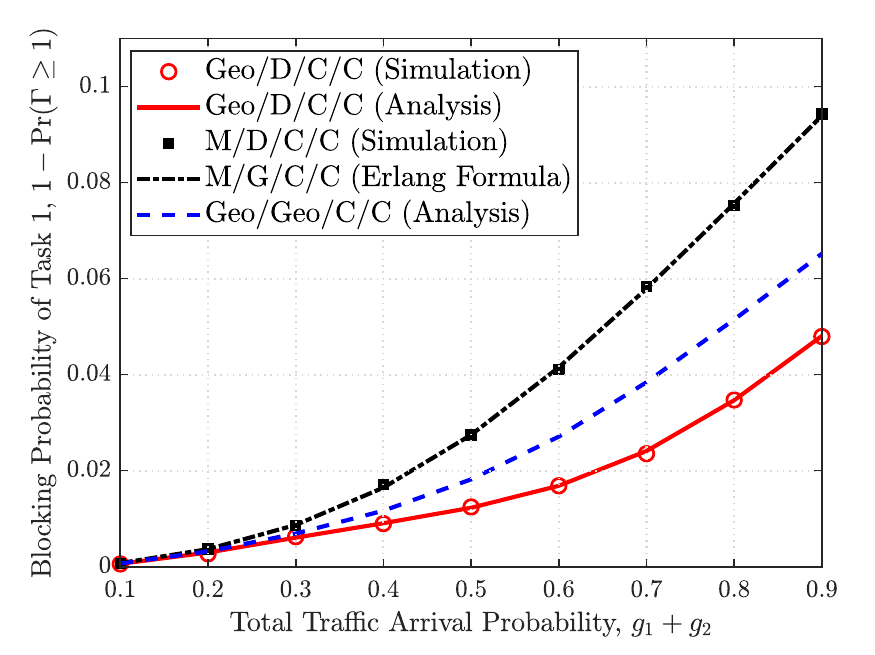}}
    \centering
    \subfloat[Task 2]
{\includegraphics[height=3.6cm,width=4.65cm]{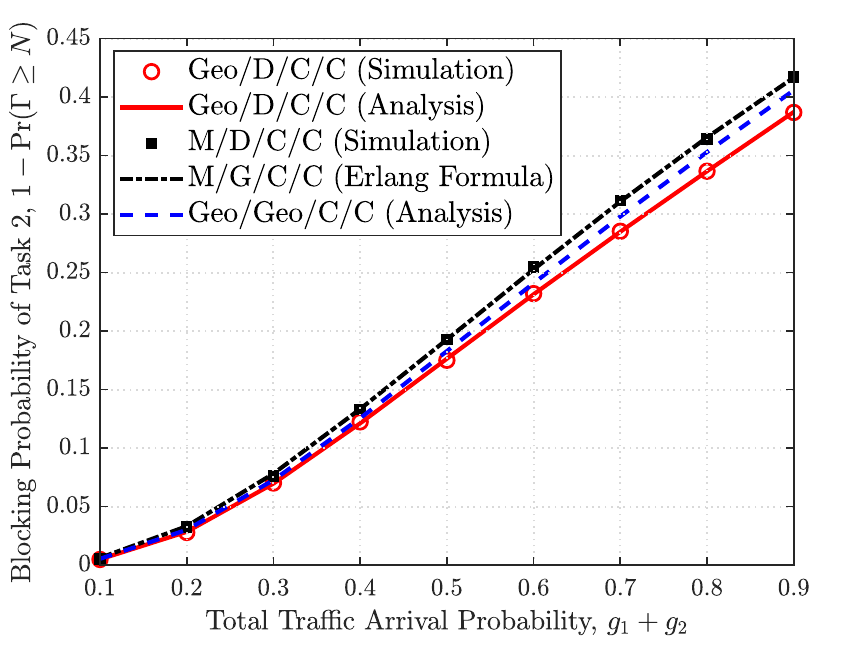}}
    \caption{Queueing model comparison. $C=12$, $N=4$, $D_{C,1}=5$, $D_{C,2}=10$,
    $\eta_1=\eta_2=1$,
    $p_{u,i}=1$, $g_1/g_2=4$. }
    \label{Fig:AoAvsC}
    \vspace{-0.5cm}
\end{figure}

\subsection{Optimization Problem Formulation}

In the considered WNCS, the ultimate goals are to optimize CoMA caused by packet drops and the timeliness of the regular task streams. Thus, we formulate the following bi-objective optimization problem:
 \begin{align}
  \mathscr{P}1:&\underset{P_{T,1}, P_{T,2},\eta_1,\eta_2}{\mathrm{minimize}}~~\{ \bar{\Delta}_1, \mathcal{C}_{\mathrm{MA}} \},\\
    &\mathrm{subject~to}~~g_1\eta_1 P_{T,1}+g_2 \eta_2 P_{T,2}\leq \bar{P},\label{eq:EConstraint}\\
    &~~~~~~~~~~~~~~\eta_1,\eta_2\in(0,1],~P_{T,1},P_{T,2}>0. \label{eq:ParaConstraint}
 \end{align}
The constraint \eqref{eq:EConstraint} ensures that the long-term average transmit power of the sensor does not exceed the prescribed power budget, and constraint \eqref{eq:ParaConstraint} specifies the physical constraints on the decision variables.

Without resorting to algorithmic simplifications of the optimization problem, we use a search method to verify whether applying a differentiated strategy to heterogeneous semantic data can provide performance improvements for the entire system, as discussed in the next section.

\section{Numerical Results and Discussion}
In this section, we validate the theoretical analysis through simulations and discuss the impact of parameters on the time-average AoAI of each task and the CoMA. Unless explicitly stated, we assume $C=8$, $N=4$,  $\omega_1=1$, $\omega_2=10$, $D_{C,i}=10~\text{slots}$, $D_{T,i}=0.1~\text{slots}$, $\sigma^2=-80~\mathrm{dBW}$, $\bar{\gamma}=5~\mathrm{dB}$, $m=1$, $\alpha_u\!=\!3$, $d_u\!=\!50$, $\bar{P}\!=\!0.18~W$, simulation time is $10^6$ slots.

Fig. \ref{fig:3} illustrates the performance metrics versus uplink access probability of task 1 $\eta_1$. Increasing $\eta_1$ degrades AoAI of task 2 and CoMA. The task-level performance analysis clearly reveals the resource contention among heterogeneous task flows. This confirms that task-specific metrics provide a more accurate characterization than aggregated AoI. Moreover, the stochastic-service approximation is accurate for Task 1 and for CoMA over the considered operating range, but its accuracy deteriorates for the AoAI of Task 2. This is because multi-server tasks are more sensitive to the temporal structure of resource releases, while AoAI depends on both the execution rate and the variability of successful actuations. The Erlang formula therefore serves as a tractable reference rather than a universally conservative approximation.

The goal in the considered WNCS is to minimize the CoMA caused by task command drops, while guaranteeing the execution timeliness of the regular task. The Pareto frontier (based on Geo/Geo/C/C approximation) in the $\{\bar{\Delta}_1, \mathcal{C}_{\mathrm{MA}}\}$ objective space, obtained by optimizing over $(P_{T,1},P_{T,2},\eta_1, \eta_2)$ is presented in Fig. 4. The results reveal a significant trade-off between $\bar{\Delta}_1$ and $\mathcal{C}_{\mathrm{MA}}$ when we employ differentiated resource allocation. Shifting toward the bottom-right of the frontier illustrates how the system trades off the AoAI of regular tasks for improved critical-task reliability. Furthermore, the traditional baseline can only adjust the effective offered loads of the two tasks simultaneously. An overly aggressive update policy intensifies contention for computational resources and degrades the execution reliability, whereas an overly conservative policy causes both tasks to suffer from insufficient update opportunities. Consequently, an optimal operating point may exist to minimize CoMA. The gap between Pareto frontier and baseline presents that differentiated allocation expands the achievable trade-off region and can improve actuation reliability.

\begin{figure}[t]
\captionsetup[subfigure]{font=footnotesize, labelfont=sf, textfont=sf}
    \centering
     \subfloat[AoAI]
{\includegraphics[height=3.6cm,width=4.8cm]{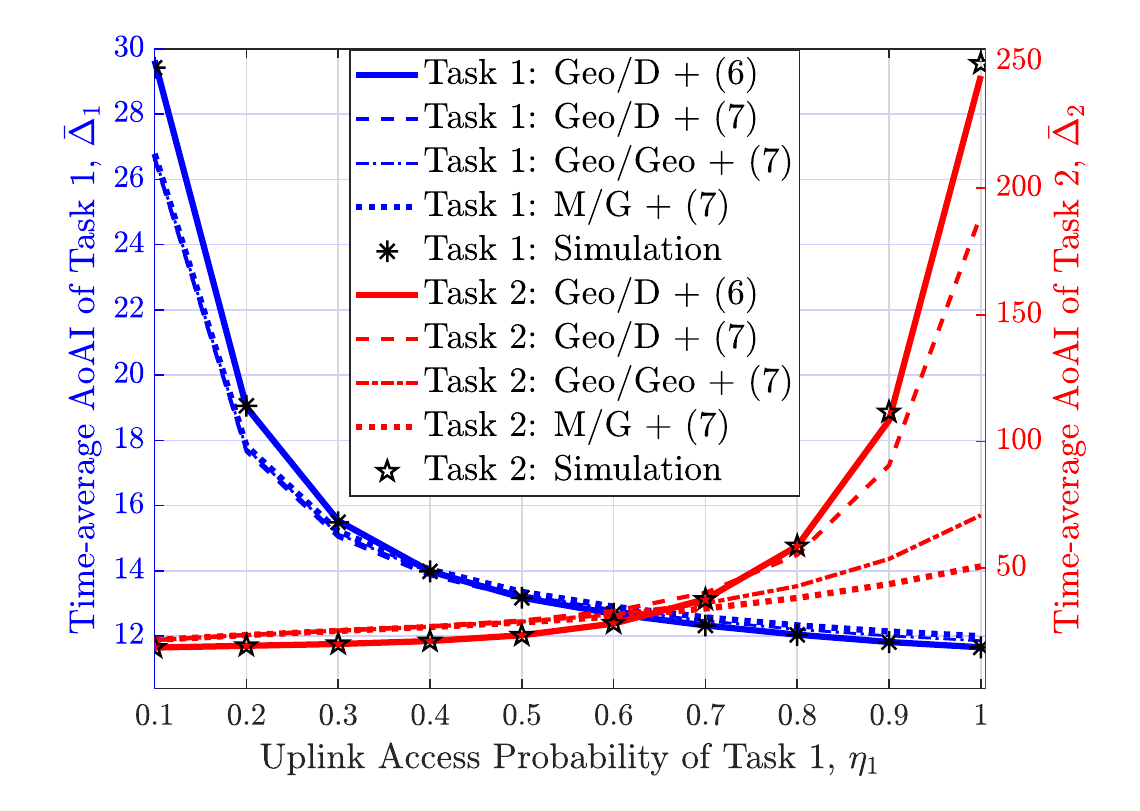}}
\centering
    \subfloat[CoMA]
{\includegraphics[height=3.6cm,width=4.5cm]{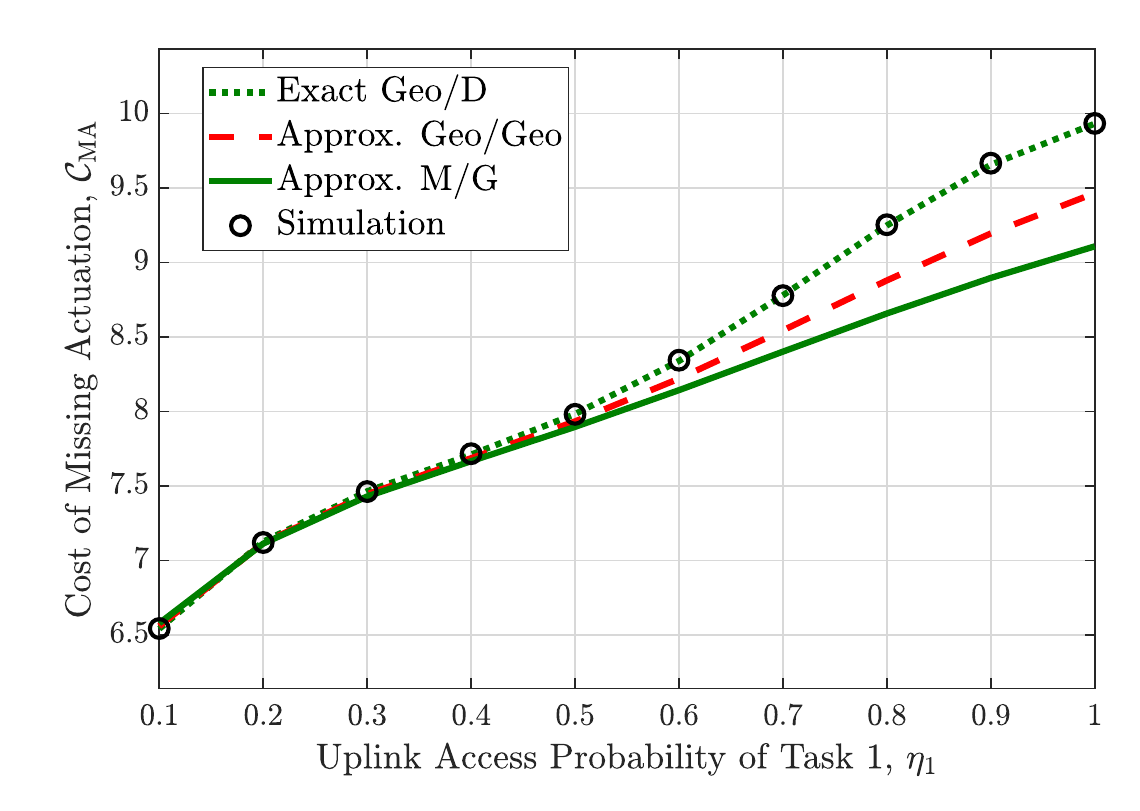}}
    \caption{Performance metrics versus uplink access probability of task 1. $P_{T,1}=0.05~W$, $P_{T,2}=0.2~W$, $\eta_2=0.8$, $g_1=0.8$, $g_2=0.2$.}
    \label{fig:3}
    \vspace{-0.5cm}
\end{figure}

\begin{figure}
    \centering
    \includegraphics[width=0.65\linewidth]{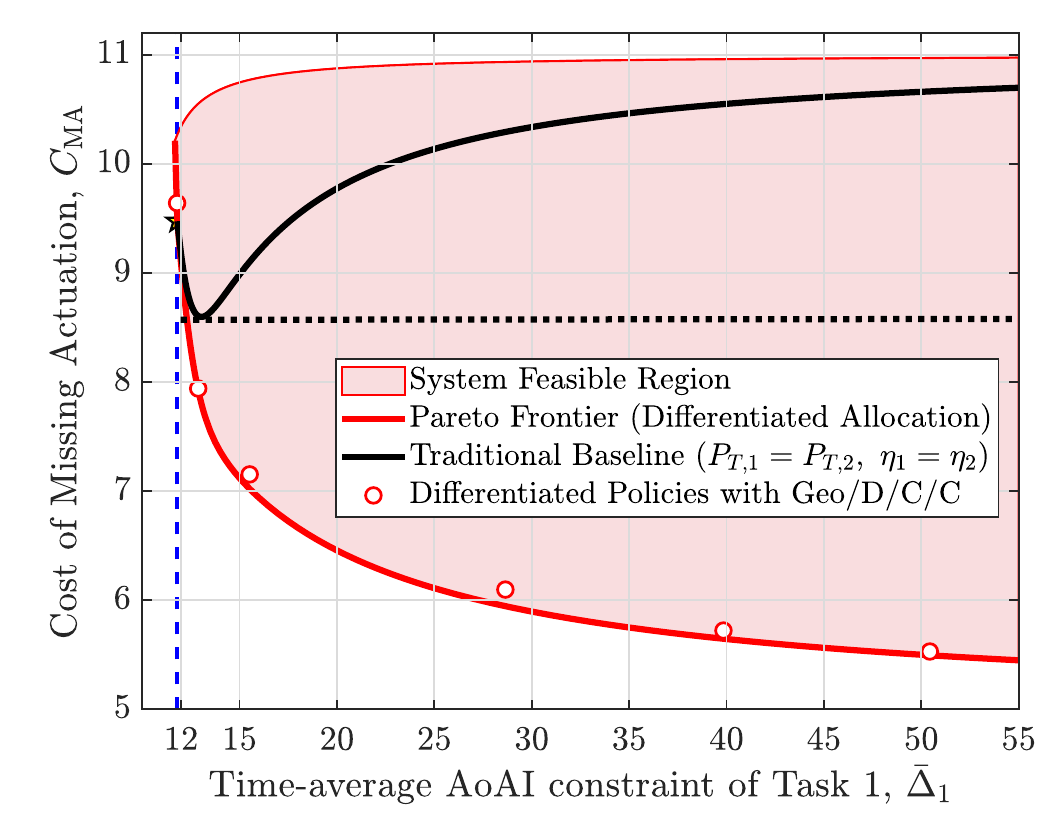}
    \caption{Joint optimization of $\bar{\Delta}_1$ and $\mathcal{C}_{\mathrm{MA}}$, $g_1=0.8$, $g_2=0.2$.}
    \label{fig:4}
    \vspace{-0.5cm}
\end{figure}

\appendices
\section{}\label{proof:AppendixA}
Let $n_1$ and $n_2$ denote the number of active Task 1 and Task 2 instances. Under the constraint $n_1+Nn_2\le C$, we partition the state space into levels based on $n_1\in\{0, 1, \dots, C\}$. Let $\boldsymbol{S}_{n_1}$ be the steady-state probability vector for Level $n_1$, forming the global vector $\mathcal{S}=[\boldsymbol{S}_0, \boldsymbol{S}_1, \dots, \boldsymbol{S}_C]$. Since a maximum of one task 1 can arrive per time slot, transitions to higher levels $k>n_1$ are strictly bounded by $k\le n_1+1$. This restricts the global balance matrix $\boldsymbol{Q}=\boldsymbol{P}-\boldsymbol{I}$ to a block Lower Hessenberg form. The global balance equation $\boldsymbol{S}\boldsymbol{Q}=\mathbf{0}$ yields the level-wise equations for $k\in\{1, 2, \dots, C\}$:
\vspace{-0.2cm}
\begin{equation}\label{eq:20}
\boldsymbol{S}_{k-1} Q_{k-1, k} + \boldsymbol{S}_k Q_{k,k} + \sum_{n_1=k+1}^{C} \boldsymbol{S}_{n_1} Q_{n_1,k} = \mathbf{0}. 
\end{equation}
Using the matrix-geometric method, we assume there exist rate matrices $R_k$ such that $\boldsymbol{S}_{n_1} \!=\! \boldsymbol{S}_k \prod_{m=k+1}^{n_1} R_m,~~ \forall n_1\!\ge\!k+1.$ 
Substituting this relationship back into the level-wise equation and factoring out $\boldsymbol{S}_k$ defines the matrix $\tilde{Q}_k$:
\begin{equation}
\tilde{Q}_k = Q_{k,k} + \sum_{n_1=k+1}^{C} \left( \prod_{m=k+1}^{n_1} R_m \right) Q_{n_1,k}. 
\end{equation}
This matrix $\tilde{Q}_k$ models a censored Markov chain, effectively folding all excursion paths from higher levels directly into the local transition of level $k$. Eq. \eqref{eq:20} then simplifies to $\boldsymbol{S}_{k-1} Q_{k-1, k} + \boldsymbol{S}_k \tilde{Q}_k = \mathbf{0}$, which defines the rate matrix $R_k = -Q_{k-1, k} \tilde{Q}_k^{-1}$, with $\boldsymbol{S}_k = \boldsymbol{S}_{k-1} R_k$. This establishes a backward recursion initialized at $k=C$ (where $\tilde{Q}_C = Q_{C,C}$) and proceeds downward to $k=1$. Finally, the base vector $\boldsymbol{S}_0$ is obtained by applying the geometric substitution:
\begin{equation}
 \boldsymbol{S}_0 \left[ Q_{0,0} + \sum_{n_1=1}^{C} \left( \prod_{m=1}^{n_1} R_m \right) Q_{n_1,0} \right] = \mathbf{0},    
\end{equation}
alongside the normalization condition $\sum_{n_1=0}^{C} \boldsymbol{S}_{n_1} \mathbf{1} = 1$.
\bibliography{ref}
\end{document}